\documentclass[runningheads]{llncs}
\usepackage[T1]{fontenc}
\usepackage{graphicx}
\usepackage{amssymb}

\usepackage{amsmath,amsfonts,bm}

\def\1{\bm{1}}

\makeatletter
\newcommand{\ve}{\@ifnextchar\bgroup{\velong}{{\bm{e}}}}
\newcommand{\velong}[1]{{\bm{#1}}}
\makeatother

\DeclareMathAlphabet{\mathsfit}{\encodingdefault}{\sfdefault}{m}{sl}
\SetMathAlphabet{\mathsfit}{bold}{\encodingdefault}{\sfdefault}{bx}{n}

\newcommand{\eat}[1]{}

\newtheorem{assumption}{Assumption}

\begin{document}
\title{A Unified Framework of Multi-Stage Multi-Winner Voting: An Axiomatic Exploration}
\titlerunning{Axioms Study for Multi-Stage Multi-Winner Voting}
%
\author{Shengjie Gong\inst{1} \and
Lingxiao Huang\inst{2} \and
Shuangping Huang\inst{1} \and
Yuyi Wang\inst{3} \and
Zhiqi Wang\inst{4} \and
Tao Xiao\inst{5} \and
Xiang Yan\inst{5} \and
Chunxue Yang\inst{5}
}
\authorrunning{Gong et al.}
%
\institute{South China University of Technology \and
Nanjing University \and
CRRC Zhuzhou Institute \and
Shanghai University of Finance and Economics \and
Huawei Taylor Lab}
\maketitle              
\begin{abstract}
Multi-winner voting plays a crucial role in selecting representative committees based on voter preferences. 
Previous research has predominantly focused on single-stage voting rules, which are susceptible to manipulation during preference collection. 
In order to mitigate manipulation and increase the cost associated with it, we propose the introduction of multiple stages in the voting procedure, leading to the development of a unified framework of multi-stage multi-winner voting rules. 
To shed light on this framework of voting methods, we conduct an axiomatic study, establishing provable conditions for achieving desired axioms within our model. 
Our theoretical findings can serve as a guide for the selection of appropriate multi-stage multi-winner voting rules.
\end{abstract}
\section{Introduction}
\label{sec:intro}

The problem of multi-winner voting is to select a winning committee of size $k$ from $m$ candidates by $n$ voters with individual preferences of candidates, which has various applications in political, social, or business settings~\cite{Skowron16:Finding,lu2011budgeted}, e.g., parliamentary elections, the selection of awards judging committee, and movie selection on the front page of Netflix. 
Due to the importance, there is a large body of study for multi-winner voting, including axiomatic study~\cite{elkind2017properties}, computational complexity~\cite{fishburn2004approval,Procaccia08:Complexity}, and approximate optimal committee selection~\cite{lu2011budgeted,Skowron15:Achieving}.

Given the wide-ranging applications of multi-winner voting, there has been a growing concern regarding potential manipulations aiming at influencing the election results in favor of certain individuals through the misrepresentation of their preferences. 
A promising approach to address this concern is the adoption of multi-stage voting protocols, which have proven effective in enhancing computational resistance to manipulation~\cite{bartholdi1991single,davies2012eliminating,narodytska2013manipulating}. 
However, there are more desired properties, known as axioms in literature, besides manipulation-proofness when choosing voting rules for specific scenarios.
Intuitively, if a single-stage voting rule satisfies an axiom, one may wonder whether the corresponding multi-stage voting rule also does.

In this paper, we develop the axiomatic study for the multi-stage multi-winner voting.
It is the initial phase of a theoretical investigation into multi-stage multi-winner voting, aiming to enhance our understanding of its design principles. 
Notably, we place less emphasis on the computational aspects of multi-stage voting in this study, as the outcomes may either yield negative results as extensions of existing NP-hard proofs or positive results that aggregate established single-stage rules. 
Instead, we focus on determining whether and under what conditions the specific desired axioms or properties can be satisfied by multi-stage multi-winner voting mechanisms, which holds significant importance for the practical selection of voting rules.

\subsection{Technique Contributions}
\label{sec:contribution}

We first introduce a unified framework of multi-stage multi-winner voting rules, denoted as $\mathcal{R}=(R_1,\ldots,R_t)$ ($t\geq 1$), where different multi-winner voting rules $R_r$ can be employed in distinct stages $r\in [t]$, and the winning committee of the $r$-th stage serves as the candidate pool for the $(r+1)$-th stage (Definition~\ref{def: mrmw}). 
Our approach allows for the capture of classic multi-winner voting rules, most importantly the score-based rules (Definition~\ref{def: abgrule}), which generalizes existing multi-stage voting rules.

We then establish a set of desired axioms for multi-winner voting and provide sufficient conditions for the satisfaction or violation of these axioms within our model. 
Specifically, we demonstrate that multi-stage score-based voting maintains solid coalition (Section~\ref{sec:score_preserve}) but breaches committee monotonicity, candidate monotonicity, and consistency (Section~\ref{sec: violating score}). 
These findings offer valuable insights for selecting appropriate rules in multi-stage voting, recommending the inclusion of rules that satisfy solid coalitions, such as SNTV, in each stage (Section~\ref{sec:discussion}).

\subsection{Other Related Work} 
\label{sec:related}


A large body of literature studies axioms of multi-winner voting.
The closest work to ours is~\cite{elkind2017properties}, which considered multiple axioms including 
consistency, 
adapted from the single-winner setting; committee monotonicity
\cite{barbera2008choose}, which ensures that the selected committee can be extended without removing anyone from it when increasing the target committee size; solid coalitions, 
which are weaker than but reflect the same idea as Dummett’s proportionality \cite{dummett1984voting} or the Droop proportionality criterion \cite{woodall1994properties}.  
Other axioms that are theoretically important but not discussed in~\cite{elkind2017properties} include Pareto efficiency~\cite{pareto1919manuale} and justified representation~\cite{Faliszewski2017multiwinnerVA}.
Some axiomatic studies of multi-stage single-winner voting have been done in \cite{narodytska2013manipulating}.
In this paper, we study the above axioms in multi-stage settings and extend to the multi-winner case.
%

The design of multi-winner voting rules has also attracted a lot of attention for decades.
Elkind et al.~\cite{elkind2017properties} suggested two natural ways to identify the similar internal structure between many of the known multi-winner rules: best-$k$ rules such as SNTV and $k$-Borda, and committee scoring rules such as Bloc and CC, which is a subclass of the score-based voting rules defined in Section \ref{sec:rule}. 
%
They picked ten single-round voting rules as examples of different ideas on (score-based) multi-winner elections, and studied whether these rules satisfy the axioms aforementioned. 
There are some commonly used rules that have not been discussed in~\cite{elkind2017properties}, such as  
approval-based rules~\cite{aziz2017justified} and Condorcet committee-related rules~\cite{Faliszewski2017multiwinnerVA}.
In this work, we extend their analysis for score-based rules to the multi-stage case and add a discussion of approval-based rules.
%



The axiomatic study for multi-stage voting is started by Smith~\cite{smith1973aggregation}, who considers rules of successive candidate elimination in a single-winner setting and showed that all scoring runoff rules fail monotonicity. 
Narodytska and Walsh~\cite{narodytska2013manipulating} studied two-stage single-winner voting rules, corresponding axioms, and their computing complexity. 
They showed that a two-stage voting rule offers advantages by inheriting appealing properties from the two stages. 
For instance, the Black process inherits the Condorcet consistency from the first stage and the properties of monotonicity, participation, and relevance to Condorcet losers from the second stage. 
On the other hand, the vulnerability to manipulation and control can be seen as an undesirable characteristic of two-stage voting rules. 
Davies et al.~\cite{davies2012eliminating} considered a model similar to ours except that they use scoring vectors and only allow a single winner. 
Their work shows that the process of sequential elimination of candidates is often considered a means to make manipulation computationally challenging. 
Borodin et al.~\cite{borodin2019primarily} focused on the primary elections within political parties, followed by a general referendum. 
Their paper points out that, in the real world, electoral and decision-making processes are often more complex, involving multiple stages. 
In this paper, we propose a unified framework and provide a systematic analysis for the axioms in multi-stage voting.



\section{Unified Framework for Multi-Stage Multi-Winner Voting}
\label{sec:model}

In this section, we establish the framework for multi-stage multi-winner voting.
We denote an election by $E = (C, V)$, which consists of a set $C$ of $m$ candidates and a group $V$ of $n$ voters. 
Before the discussion on multi-stage rules, we revisit the definition of single-stage multi-winner voting rules.
\begin{definition}[\bf Multi-Winner Voting Rules]
\label{def:mult-winner}
A multi-winner voting rule, denoted as $R$, is a function that, given an election $E = (C,V)$ and a positive integer $k$ $(1\leq k< m)$, outputs a collection $R(E,k)$ of possible winning committees $S\subseteq C$ of size $k$.
\end{definition}
%
%
We now introduce a concept of multi-stage multi-winner voting rules, wherein each stage employs a multi-winner voting rule (which may vary across stages), and the winning committee from each stage serves as the candidate set for the subsequent stage.
\begin{definition}[\bf $t$-Stage Multi-Winner Voting Rules]
\label{def: mrmw}
%
A $t$-stage multi-winner voting rule, denoted as $\mathcal{R} = (R_1, R_2, \ldots, R_t)$, is defined as a function that, given an election $E=(C,V)$ and an integer vector $\vec{v} = (k_1, k_2, \ldots, k_t)$ where $m > k_1 > k_2 >...> k_t$, outputs a collection $\mathcal{R}(E,\vec{v})$ of possible winning committees $S\subseteq C$ of size $k_t$. 
This output is contingent on the existence of a sequence $(S_0,S_1,\ldots,S_t)$ that satisfies the following conditions: (\romannumeral1) $S_0 = C$ and $S_t = S$; and (\romannumeral2) $S_r\in R_r(E_{r},k_r)$ for each $r\in[t]$, where $E_{r} = (S_{r-1}, V)$.
\end{definition}
Notice that a multi-stage multi-winner voting rule can be considered a special case within the scope of Definition~\ref{def:mult-winner}.
This implies that multi-stage voting represents one among various methodologies for determining the final winning committees.
The sequence $(S_0, S_1, \ldots, S_t)$ in a $t$-stage voting rule may be interpreted as a trajectory of successive shortlists, that certificates $S_t$ as the ultimate winning committee.
%
%
%
%
%
%

A multi-winner voting rule determines its output according to the voters' preferences for candidates.
We assume that a voter $v$’s preference is represented by a ranking of the candidates.
%
For each candidate $c$, we define $p_v(c) = l$ to indicate that candidate $c$ is ranked as the $l$-th favorite candidate by voter $v$.
Thus $(p_v(c_1), ...,p_v(c_m))$ forms a permutation of $[m]$\footnote{In this paper, we assume that there is no tie in the order of preferences.}.
%
For multi-stage voting rules, we need to establish the ranking of voters specifically on the candidate set $S_{r-1}$ during the $r$-th stage instead of the complete set $C$.
%
%
\begin{assumption}[\bf Rankings are preserved within any subset]
    \label{assumption:ranking_ordeer}
    For any subset $S\subseteq C$, let $p^S_v$ represent the ranking of $v$ on $S$. 
    We assume that for any pair of candidates $c, c'\in S$, $p^S_v(c) < p^S_v(c')$ holds if and only if $p_v(c) < p_v(c')$.
\end{assumption}

\subsection{Score-Based Rules}
\label{sec:rule}
We introduce a unified framework of multi-stage multi-winner voting rules, known as \emph{score-based rules} (Definition~\ref{def: abgrule}). 
It encompasses a wide range of important multi-winner voting procedures that has been well studied in the literature (e.g.,~\cite{chamberlin1983representative,elkind2017properties}).
We begin by defining single-stage score-based rules before extending them to the multi-stage context.
%
%
%

Two parameters characterize a single-stage score-based rule: first, how a voter assigns scores to individual candidates; and second, how a voter evaluates a committee comprised of candidates that she supports in different degree.
Score-based rules operate under the assumption that a voter $v$ assigns a score $\gamma^{m,k}(p_v(c))$ to each individual candidate $c$, and this score is non-increasing based on $c$'s position $p_v(c)$ in $v$'s preference list. 
The score $\gamma^{m,k}(\cdot)$ may depend on the size $m$ of the candidate set and the size $k$ of the winning committees. This is because a voter's preference for candidates may vary with changes in $m$ and $k$, influencing the determination of position scores.
As an illustration, consider the approval score, wherein a uniform score is assigned to the candidates ranking within the top $k$ positions. 
Another example is the Borda score, which assigns a score of $m-p$ to the candidate occupying the $p$-th position in a voter's preference list. 
To prepare for implementing a position score function in multi-stage voting, where the sizes of the candidate pool and the target committee vary in different stages, it is imperative to precisely define the score function for distinct cases of $(m,k)$.

We allow two distinct forms for a voter's assigned score to a committee: it can be either the cumulative sum of scores assigned to individual candidates within the committee, or the score specifically designated for the voter's preferred candidate within the same committee.
Let us consider the committee's score as a norm $\beta$ applied to the score vector assigned to all candidates in the committee. When $\beta=\ell_1$, each candidate $c$ in a committee $S$ contributes a utility of $\gamma^{m,k}(p_v(c))$ to voter $v$.
These rules with $\beta=\ell_1$ are also referred to as \emph{weakly separable committee scoring rules} in \cite{elkind2017properties}. 
In the case of $\beta = \ell_{\max}$, a voter $v$ evaluates a committee $S$ based solely on one representative candidate within $S$. 
With the parameters $\beta$ and $\vec{\gamma}$, we formally define score-based rules.
%
%
%
%
%
\begin{definition}[\textbf{Score-based rules; $(\beta, \vec{\gamma})$-rule}]
\label{def: abgrule}
A score-based rule is parameterized by a norm $\beta\in\{\ell_1, \ell_{\max}\}$ and an infinite-dimensional vector function 
$$\vec{\gamma}=\left(\gamma^{2,1},\gamma^{3,1},\gamma^{3,2},...,\gamma^{m,k},...\right),$$
where $\gamma^{m,k}:[m]\rightarrow \mathbb{R}$ is a non-increasing position score function.
%
Given an election $E = (C, V)$ and a target committee size $k$, the score $f_v(S)$ of a committee $S$ given by a voter $v\in V$ is defined as
$$f_v(S)=\left\{
\begin{aligned}
 \sum_{c\in S}\gamma^{m,k}(p_v(c)) & \quad \text{if}\ \beta = \ell_{1};\\
 \max_{c\in S}\gamma^{m,k}(p_v(c)) & \quad \text{if}\ \beta = \ell_{\max}.
\end{aligned}
\right.$$
The $(\beta,\vec{\gamma})$-rule voting returns the committees with maximum scores of $\sum_{v\in V}f_v(S)$ over all possible $S$ of size $k$.
\end{definition}
\paragraph{Examples of score-based rules.}
The class of score-based rules encompasses multiple commonly-used multi-winner voting rules; summarized as follows.
\begin{itemize}
\item Single Non-Transferable Vote (SNTV): the $(\ell_1, \textsc{Plu})$-rule, where $\textsc{Plu}$ is the plurality score $\textsc{Plu}^{m, k}(p) = \mathbb{I}_{\{p = 1\}}$.
\item Bloc: the $(\ell_1, \textsc{App})$-rule, where $\textsc{App}$ represents the approval score $\textsc{App}^{m, k}(p) = \mathbb{I}_{\{p \leq k\}}$.
\item Borda: the $(\ell_1, \textsc{Borda})$-rule, where $\textsc{Borda}$ is the Borda score $\textsc{Borda}^{m, k}(p) = m - p$.
\item Chamberlin-Courant (CC): the $(\ell_{\max}, \textsc{Borda})$-rule.
\end{itemize}
%
We can extend score-based rules introduced in Definition~\ref{def: abgrule} to the multi-stage case (Definition~\ref{def: mrmw}). Specifically, we define $\mathcal{R}$ as a multi-stage $(\beta, \vec{\gamma})$-rule if it employs the $(\beta, \vec{\gamma})$-rule in each stage of the voting process.
This framework for multi-stage multi-winner voting consists of a broad spectrum of previously defined voting procedures.
For instance, \emph{Single Transferable Vote (STV)} can be regarded as a $(m-1)$-stage rule with a vector $\vec{v} = (m-1, m-2, \dots, 1)$, where each stage employs a $(\ell_1, \textsc{Plu})$-rule. In practical terms, this represents that STV eliminates one candidate with the minimum plurality score at each stage.
Another example within this framework is \emph{Baldwin's rule} (refer to~\cite{narodytska2013manipulating}). Baldwin's rule is also structured as a $(m-1)$-stage rule with $\vec{v} = (m-1, m-2, \dots, 1)$, but each stage utilizes a $(\ell_1, \textsc{Borda})$-rule. 

Notice that Definition~\ref{def: mrmw} enables the utilization of different position score functions $\vec{\gamma}$ across various stages. However, for the sake of simplicity, we do not delve into the exploration of such general multi-stage voting rules within the scope of this paper. The investigation of employing distinct $\vec{\gamma}$-scores in different stages is an intriguing direction for future research.

    %
    %
    %
    %
    
    %
    %

\subsection{Axioms Related to Score-Based Rules}
\label{sec:axiom_score}

Next, we present some axioms that are desirable for score-based rules. 
The first axiom is \emph{Solid Coalition} representing an implementation of proportionality, which is a notion that requires a voting rule to select a committee representing as precisely as possible the opinions of the society. A typical example of scenarios where proportionality is attached importance is parliamentary elections~\cite{sanchez2017monotonicity}.
\begin{definition}[\bf Solid Coalition]
\label{def:solid_coalition}
A score-based rule $R$ satisfies Solid Coalition iff $\forall E = (C, V)$ and $k\in [m]$, if at least $\frac{n}{k}$ voters rank some candidate $c$ first then $c$ belongs to every committee in $R(E,k)$.
\end{definition}
\eat{
\emph{Pareto efficiency}. It requires that a dominated committee must never be output.
When the goal of a multi-winner rule is to select the “best” committee, Pareto efficiency is often considered to be a minimal requirement.
\begin{definition}[\textbf{Pareto efficiency for scored-based rules}]
A committee $S_1$ dominates a committee $S_2$ if
(\romannumeral1) each voter scores $S_1$ at least as high as $S_2$ (for $v\in V$ it holds that $f_v(S_1) \geq f_v(S_2)$), and (\romannumeral2) there is one voter with a strictly higher score (there exists $u\in V$ with $f_u(S_1) > f_u(S_2)$).
A multi-winner score-based rule $R$ satisfies Pareto efficiency if $R$ never outputs dominated committees.
\end{definition}
}

Our next axiom, \emph{Committee Monotonicity}, requires that a target committee size increase should not result in the elimination of any of the currently selected candidates. 
This requirement arises when we are looking for the “best” candidates.
An example is a hiring process in which the number of candidates is not determined beforehand.
A committee monotone rule actually produces a ranking of candidates, indicating which one should be hired for an extra position~\cite{lackner2020multi}.
%
%
%
%
\begin{definition}
[\bf Committee Monotonicity]
\label{def: comm_mono_multi}
Given an integer $t\geq 1$, a $t$-stage multi-winner voting rule $\mathcal{R}$ is committee monotone if, for any election $E = (C, V)$ and $t$-dimensional vectors $\vec{v}_1 = (k^1_1, k^1_2,\ldots, k^1_t)$ and $\vec{v}_2 = (k^2_1, k^2_2,\ldots, k^2_t)$ with $k^1_t+1 = k^2_t$ in the last dimension, we have (\romannumeral1) if a committee $S \in \mathcal{R}(E,\vec{v}_1)$ then there exists a $S' \in \mathcal{R}(E,\vec{v}_2)$ such that $S \subset S'$; (\romannumeral2) if $S \in \mathcal{R}(E,\vec{v}_2)$ then there exists a $S' \in \mathcal{R}(E,\vec{v}_1)$ such that $S' \subset S$.
\end{definition}
%
%

The following two axioms, \emph{Candidate Monotonicity} and \emph{Consistency}, are generally desirable if one is interested in rules that are fair to candidates.
As Janson~\cite{janson2016phragmen} points out, no perfect election exists, but particularly disturbing are cases when changing some votes in favor of an elected candidate may result in her losing the election. 
Candidate Monotonicity is essential when candidates are not inanimate objects.
\begin{definition}[\bf Candidate Monotonicity]
\label{def:candidate_monotonicity}
A score-based rule $R$ is candidate monotone iff $\forall E = (C, V), k\in [m]$ and $c\in C$, if $c\in S$ for some $S\in R(E, k)$, then if we shift $c$ one position forward in a voter $v$ and obtain an election $E'$, we have $c\in S'$ for some $S'\in R(E', k)$.
\end{definition}
\noindent Consistency requires that if two disjoint groups of voters have the same election result, a voting rule should also arrive at this outcome if the two groups are united.
\begin{definition}[\bf Consistency]
\label{def:consistency}
\sloppy
A multi-winner voting rule $R$ is consistent iff $\forall E_1(C, V_1), E_2(C, V_2)$ and $k\in [m]$, if $R(E_1, k)\cap R(E_2, k)\neq \varnothing$, then we have $R((C, V_1 + V_2), k) = R(E_1, k)\cap R(E_2, k)$.
\end{definition}

\paragraph{\textbf{Remark on Axiom Selection.}}
The reason why we choose the aforementioned axioms for our study is based on the belief that a "good" multi-stage multi-winner rule has distinct meanings when applied in different scenarios. Firstly, the objective may be to select the "optimal" committee, necessitating the consideration of Committee Monotonicity (Definition~\ref{def: comm_mono_multi}). 
    However, in other cases, fairness assumes greater importance. In such instances, we typically examine the axioms of Solid Coalitions (Definition~\ref{def:solid_coalition}), 
    aiming for the elected members to represent the opinions of voters, thereby prioritizing fairness to voters. If our focus shifts towards fairness to the candidates and the need for the multi-stage rule to be easily understandable by them, we then take Candidate Monotonicity (Definition~\ref{def:candidate_monotonicity}) 
    and Consistency (Definition~\ref{def:consistency}) into account.



\section{Axiomatic Study for Multi-Stage Score-Based Rules}
\label{sec: score}
In this section, we present the axiomatic study for score-based rules. 
We first demonstrate that if a set of single-stage rules satisfies Solid Coalition, this axiom can be maintained in multi-stage voting scenarios, where each stage employs one of these single-stage rules.
Subsequently, we provide election examples to illustrate that the axioms of Committee Monotonicity, Candidate Monotonicity, and Consistency are violated in the multi-stage setting.
\subsection{Preserving Axioms in Multi-Stages}
\label{sec:score_preserve}
The following theorem suggests that when employing rules with the Solid Coalition property at each stage, the multi-stage voting rule also satisfies the Solid Coalition property.
For example, we use SNTV rule in each stage of STV procedure, and SNTV satisfies the Solid Coalition \cite{elkind2017properties}.
Consequently, STV inherits this property throughout its multi-stage process.
\begin{theorem}[\bf Solid Coalition preserves in multi-stage voting]
    \label{thm:solid coalitions}
    Let $t\geq 1$ be an integer and $\mathcal{R} = (R_1, R_2, \ldots, R_t)$ be a $t$-stage multi-winner voting rule.
    If $R_r$ satisfies Solid Coalition for each $r\in [t]$, $\mathcal{R}$ also satisfies Solid Coalition.
\end{theorem}
\begin{proof}
    %
    %
    %
    %
    %
    %
    %
    %
    %
    %
    Fix any $E=(C,V)$ and vector $\Vec{v}=(k_1,k_2,\dots,k_t)$.
    %
    Suppose that there are at least $\frac{n}{k_t}$ voters ranking a candidate $c$ first. 
    Due to the fact that $R_1$ satisfies the Solid Coalition property and that at least $\frac{n}{k_t} \geq \frac{n}{k_1}$ voters rank $c$ first,
    %
    %
    $c$ must be one of the winners in $R_1$. 
    We observe that the number of voters who assign the highest rank to $v$ will not decrease in the successive stages of voting.
    %
    Therefore, we can similarly show that $c$ appears in each stage's winning committees, which completes the proof.
    %
    %
    \qed
\end{proof}

\subsection{Violating Axioms in Multi-Stages}
\label{sec: violating score}
In this section, we discuss the axioms including Committee Monotonicity, Candidate Monotonicity, and Consistency.

%
%
%
%

Before presenting the main theorem in this section, we first clarify the non-triviality for multi-stage $(\beta,\Vec{\gamma})$ rules.
%
%
%
%
%
For a multi-stage $(\ell_1,\Vec{\gamma})$-rule, any stage of voting that applies a position score function $\gamma^{m,k}$ such that $\gamma^{m,k}(1) = \gamma^{m,k}(m)$ has all the committees of size $k$ as the winning committee.
Moreover, in the case of $\beta = \ell_{\max}$, 
if $\gamma^{m, k}(1)= \gamma^{m, k}(m-k+1)$, 
then all the committees of size $k$ have the same score of $f_v(S) = \gamma^{m,k}(1)$ for each voter $v$.
These cases are deemed to be of triviality and minimal merit consideration.
%
Thus, we propose the following definition of rationality to specify the non-trivial rules.
\begin{definition}
[\bf Rationality of $(\beta, \vec{\gamma})$-rules] 
\label{def: ration_l1}
We define a multi-stage $(\beta,\vec{\gamma})$-rule as rational if, for any given pair $(m, k)$, the following conditions hold:
\begin{itemize}
    \item if $\beta = \ell_1$, $\gamma^{m, k}(1) > \gamma^{m, k}(m)$;
    \item if $\beta = \ell_{\max}$, $\gamma^{m, k}(1) > \gamma^{m, k}(m-k+1)$.
\end{itemize}
%
\end{definition}
%

%
%


We now show that the axioms of Committee Monotonicity, Candidate Monotonicity, and Consistency do not preserve in any rational multi-stage $(\beta, \vec{\gamma})$-rule.
Here we begin with the structure of our counter-example for Candidate Monotonicity.
The proof ideas for Committee Monotonicity and Consistency are analogous.
%

Consider an election process that, in the $(t-1)$-th stage, selects two candidates from a pool of three, and in the final stage, chooses one candidate from the previously selected two.
Let there be three candidates $a, b, c$ before the $(t-1)$-th stage.
The main structure that we use in the construction of voters' preferences is a cycle of $(a,b,c)$ (Definition~\ref{def: cyc}).
The characteristic of this cycle is that each candidate experiences an equal number of occurrences in every possible position.
Specifically, to form a cycle of $(a,b,c)$, we can construct a Group of three voters whose preference order are $a\succ b\succ c$, $b\succ c\succ a$, and $c\succ a\succ b$ respectively. 
%
The voter set is formed by duplicating this Group of voters multiple copies.
%
Due to this characteristic of cycles, each candidate can be the final winner. 
It can be observed that the final winning candidate will be $a$, if $c$ is eliminated in the $(t-1)$-th stage, while the final winner is $c$, if $b$ is knocked out first.
%
%
With rationality of $\vec{\gamma}$, there are two possibilities: $\gamma^{3, 2}(2) > \gamma^{3, 2}(3)$, or $\gamma^{3, 2}(1) > \gamma^{3, 2}(2)$.
In the former case, we transfer a voter with preference list $c\succ b \succ a$ to that with $c\succ a \succ b$.
If the latter, we shift $a$ one position forward in a voter with preference list $b\succ a \succ c$.
%
%
Consequently, $c$ becomes the sole winner in the last stage since $b$ becomes less important and must be eliminated in the $(t-1)$-th stage.
Therefore, the axiom of Candidate Monotonicity is violated.
%
%

The proofs for Committee Monotonicity and Consistency are based on similar constructions of such cycles.
Formally, we define the cycle of a sequence and the permutation of a voter's preference list.
\begin{definition}
[\bf The cycle of a sequence]
\label{def: cyc}
Given a sequence $s = (s_1, s_2, \ldots, s_m)$ of $m$ elements and an integer $i\in [m]$, the cycle $\rho_i(s)$ is the sequence $s'$ with:
\[
s'_j = \left\{\begin{matrix}
s_{j+i} &  j\leq m - i, \\
s_{j + i - m} & \text{otherwise}.
\end{matrix}\right.
\]
\end{definition}
\noindent
In other words, $\rho_i(s)$ shifts all elements in the sequence $s$ by $i$ positions to the right. 
If an element is shifted past the end of the sequence, it wraps around to the front of the sequence.
\begin{definition}
[\textbf{Full permutation of a sequence}]
\label{def: perm}
A permutation is a bijection of a sequence to itself.
Let $\Pi(s)$ denote the set of all permutations of sequence $s$.
\end{definition}
We use the aforementioned concepts to simplify the description of preference lists. For example, consider a candidate set $C = \{a, b, c, d, e\}$. A ranking $a\succ \rho_1(b,c,d)\succ e$ can be interpreted as $a\succ c\succ d\succ b\succ e$.

%

%

\begin{theorem}[\textbf{Committee Monotonicity does not preserve in a rational multi-stage $(\beta,\vec{\gamma})$-rule}]
\label{thm: l1l1_comm_mono}
Let $t\geq 2$ be an integer, $\beta\in\{\ell_1, \ell_{\max}\}$, and $\mathcal{R}$ be a rational $t$-stage $(\beta,\vec{\gamma})$-rule. 
Then $\mathcal{R}$ does not satisfy Committee Monotonicity.
\end{theorem}
\noindent 
According to this theorem, any multi-stage voting rule that consists of scoring-based rules does not exhibit Committee Monotonicity, even if the individual rules applied in every stage satisfy this property. 
For instance, when combining several Borda rules to form a multi-stage voting rule, the resulting multi-stage rule no longer exhibit the Committee Monotonicity, though each single-stage rule individually satisfies it.
\begin{proof}[Proof of Theorem~\ref{thm: l1l1_comm_mono}]
We construct a two-stage voting procedure, which indicates that Committee Monotonicity does not preserve in a rational multi-stage $(\beta, \vec{\gamma})$-rule.
Note that this example can be simply extended to arbitrary stages by adding dummy candidates.

For $(\ell_1, \vec{\gamma})$-rules, let $C = \{a, b, c, d, e\}, \vec{v}_1 = (4, 2)$, and $\vec{v_2} = (2, 1)$. 
Committee Monotonicity requires that for any $S\in \mathcal{R}(E, \vec{v}_2)$, there exists a $S'\in \mathcal{R}(E, \vec{v}_1)$ such that $S \subset S'$.
However, we show that it is possible to construct a set $V$ of voters such that
$\{a\}\in \mathcal{R}(E, \vec{v}_2)$ but $\mathcal{R}(E, \vec{v}_1) = \{\{b, c\}\}$, which does not include a possible winning committee $S'$ such that $\{a\} \subset S'$.
Due to the rationality assumption, either $\gamma^{4, 2}(1)>\gamma^{4, 2}(3)$ or $\gamma^{4, 2}(2)>\gamma^{4, 2}(4)$ holds.
\begin{itemize}
    \item $\gamma^{4, 2}(1)>\gamma^{4, 2}(3)$. We construct $V$ with $6$ groups of voters, each having preference lists as follows:
\begin{align*}
    &\text{Group 1:}\quad 1\times\ \rho_i(b, c, a, e)\succ d\quad \forall i\in \{1,2,3,4\};\\
    &\text{Group 2:}\quad 1\times\ \rho_i(c, b, a, e)\succ d\quad \forall i\in \{1,2,3,4\};\\
    &\text{Group 3:}\quad 1\times\ \rho_i(a, b, c, d)\succ e\quad \forall i\in \{1,2,3,4\};\\
    &\text{Group 4:}\quad 1\times\ \rho_i(a, c, b, d)\succ e\quad \forall i\in \{1,2,3,4\};\\
    &\text{Group 5:}\quad 200 \times\ \pi \succ e\quad \forall \pi\in \Pi(a, b, c, d);\\
    &\text{Group 6:}\quad 100\times\ \pi \succ d\quad \forall \pi\in \Pi(a, b, c, e).
\end{align*}

We begin our analysis with $\vec{v}_1 = (4, 2)$. In the first stage, $e$ is eliminated as a result of preference of the voters in Group 5. 
Then in the second stage, $d$ cannot win due to the presence of Group 6. 
The scores of $a, b$, and $c$ are initially equivalent. 
However, after $e$ is knocked out, the score of $a$ is less than that of $b$ and $c$ since $\gamma^{4, 2}(1)+\gamma^{4, 2}(2)>2\gamma^{4, 2}(3)$. Therefore, $\{b, c\}$ is the final winning committee.

On the other hand, consider the voting procedure that applies $\mathcal{R}$ on $(E, \vec{v}_2)$.
In the first stage, $d$ and $e$ must be eliminated because of the existence of Group 5 and Group 6. 
Further, the scores of $a, b$ and $c$ are the same, and thus $\{a, b\}$ is a possible winning committee. 
In the second stage, the scores of both candidates are also the same, so $a$ is a possible winner. 
The Committee Monotonicity is violated.

\item $\gamma^{4, 2}(2)>\gamma^{4, 2}(4)$. We modify Group 1 in the voter set $V$ constructed above as $1\times\ d\succ \rho_i(b, c, a, e)$ $\forall i\in \{1,2,3,4\}$ and Group 2 as $1\times\ d\succ \rho_i(c, b, a, e)$ $\forall i\in \{1,2,3,4\}$. 
Then the difference between the score of $a$ and $b, c$ is $\gamma^{4, 2}(2)+\gamma^{4, 2}(3)-2\gamma^{4, 2}(4) > 0$.
The previously established analysis remains applicable.
\end{itemize}

\noindent Next, we provide counter-examples for $(\ell_{\max},\vec{\gamma})$-rules.
We also construct a two-stage election. Let $C = \{a, b, c, d, e, f, g\}$ and $\vec{v}_1 = (2, 1), \vec{v}_2 = (5, 2)$. The voter set $V$ consists of 5 groups of voters:
\begin{align*}
    \text{Group 1:}& \quad 200\times\ \pi\succ g\quad\forall \pi \in \Pi(a, b, c, d, e, f);\\
    \text{Group 2:}& \quad 200\times\ \pi\succ f\quad\forall \pi \in \Pi(a, b, c, d, e, g);\\
    \text{Group 3:}& \quad 100\times\ \pi\succ e\quad\forall \pi \in \Pi(a, b, c, d, f, g);\\
    \text{Group 4:}& \quad 100\times\ \pi\succ d\quad\forall \pi \in \Pi(a, b, c, e, f, g);\\
    \text{Group 5:}
    &\quad f\succ b\succ d\succ e\succ c\succ g\succ a\quad f\succ c\succ d\succ e\succ b\succ g\succ a\\
    &\quad f\succ b\succ d\succ e\succ a\succ g\succ c\quad f\succ c\succ d\succ e\succ a\succ g\succ b\\
    &\quad d\succ a\succ f\succ g\succ b\succ e\succ c\quad d\succ a\succ f\succ g\succ c\succ e\succ b\\
    &\quad d\succ f\succ b\succ e\succ c\succ a\succ g\quad d\succ f\succ c\succ e\succ b\succ a\succ g\\
    &\quad d\succ e\succ a\succ f\succ b\succ c\succ g\quad d\succ e\succ a\succ f\succ c\succ b\succ g\\
    &\quad d\succ f\succ b\succ e\succ a\succ c\succ g\quad d\succ f\succ c\succ e\succ a\succ b\succ g.
\end{align*}
%
%
We next show that $\{a\}\in \mathcal{R}(E, \vec{v}_1)$ but $\mathcal{R}(E, \vec{v}_2) = \{b, c\}$.
In $\mathcal{R}(E, \vec{v}_1)$, it is guaranteed by Group 1-4 that candidates $d,e,f,g$ are eliminated in the first stage.
Additionally, the voters in Group 5 ensure that the score of the committee $\{a,b\}$ is greater than or equal to that of $\{b,c\}$ and $\{a,c\}$.
This is because the score of $\{a,b\}$ is $4(\gamma^{7,2}(2)+\gamma^{7,2}(3)+\gamma^{7,2}(4))$, while both $\{a,c\}$ and $\{b,c\}$ score $4(\gamma^{7,2}(2)+\gamma^{7,2}(3)+\gamma^{7,2}(5))$.
Thus $\{a,b\}$ is the winning committees in the first stage.
In the second stage, $\{a\}$ and $\{b\}$ have identical scores, therefore $\{a\}\in \mathcal{R}(E, \vec{v}_1)$.
In $\mathcal{R}(E, \vec{v}_2)$, the voters from Groups 1 - 4 ensures that $\{a,b,c,d,e\}$ is the only winning committee in the first stage.
%
%
Then in the second stage, Groups 3 and 4's voters guarantee the elimination of $d,e$.
Through calculation of the scores given by the voters in Group 5, it can be verified that both $\{a,b\}$ and $\{a,c\}$ achieve a score of $2\gamma^{5,2}(1)+4\gamma^{5,2}(2)+2\gamma^{5,2}(3)+4\gamma^{5,2}(4)$, while $\{b,c\}$ gets a score of $4\gamma^{5,2}(1)+4\gamma^{5,2}(2)+2\gamma^{5,2}(3)+2\gamma^{5,2}(4)$.
Consequently, we have $\mathcal{R}(E, \vec{v}_2) = \{\{b, c\}\}$.
\qed
\end{proof}
%
%
%

\begin{theorem}[\textbf{Candidate Monotonicity does not preserve in a rational multi-stage $(\beta,\vec{\gamma})$-rule}]
\label{thm: l1l1_cadi_mono}
Let $t\geq 2$ be an integer, $\beta\in\{\ell_1, \ell_{\max}\}$, and $\mathcal{R}$ be a rational $t$-stage $(\beta,\vec{\gamma})$-rule. 
Then $\mathcal{R}$ does not satisfy Candidate Monotonicity.
\end{theorem}
\noindent 
To provide an intuition for our analysis, consider a point runoff system (see e.g.~\cite{smith1973aggregation}), which involves the successive use of the $\ell_1$ scoring function to eliminate candidates with lower scores until a single winner is obtained. 
It can be considered as a specific instance of our framework with $\beta = \ell_1$.
Smith \cite{smith1973aggregation} demonstrated that no point runoff system involving two or more stages exhibits Candidate Monotonicity.
In Theorem \ref{thm: l1l1_cadi_mono}, we extend our analysis beyond the $\ell_1$ scoring function used in point runoff systems and include the $\ell_{\max}$ scoring function, such as Chamberlin-Courant's (CC) rules.
%
%
We reach a more comprehensive conclusion that regardless of how we combine the $\ell_1$ and $\ell_{\max}$ scoring functions, the resulting multi-stage voting rule fails to satisfy Candidate Monotonicity.
%
%
\begin{proof}[Proof of Theorem~\ref{thm: l1l1_cadi_mono}]
We present below a counter-example which works for both $(\ell_1,\vec{\gamma})$-rules and $(\ell_{\max},\vec{\gamma})$-rules.
%
%
Let there be three candidates $a,b,$ and $c$.
The target size of the winning committees in each stage is represented by $\vec{v} = (2, 1)$.
%
%
There are two groups of voters in $V$ with preferences as below.
\begin{align*}
    &\text{Group 1:}\quad 10\times\ \rho_i(c, a, b)\quad \forall i \in \{1,2,3\};\\
    &\text{Group 2:}\quad 1\times\ \pi\quad \forall \pi \in \Pi(a, b, c).
\end{align*}
We run an election that selects a winning committee of size $2$ in the first stage and a sole winner in the second stage.
If a $(\beta, \vec{\gamma})$-rule is applied in each stage, $\{a,b\}$ is one of the winning committees for the first stage and $a$ is the final winner between $a$ and $b$.

If $\gamma^{3, 2}(1) = \gamma^{3, 2}(2) > \gamma^{3, 2}(3)$ holds, we shift $a$ one position forward in $c\succ b\succ a$. If $\gamma^{3, 2}(1) > \gamma^{3, 2}(2)$, we shift $a$ one position forward in $b\succ a \succ c$.
Then $\{a, c\}$ becomes the sole winning committee in both cases under the assumption that $\vec{\gamma}$ is rational. 
In the second stage, $\{c\}$ is the final winning committee, which does not include $a$.
%

%
%

%
%
%
%
This result can be generalized to any $m \geq 3$, $t \geq 2$, and $\vec{v}=(k_1, \dots, k_{t-1}, k_t)$ with $k_t = k_{t-1}-1$.
%
Specifically, we can add $(k_{t} - 1)$ candidates to the front and $(k_1 - k_{t-1})$ to the end of each voter's preference list presented above.
Then the candidates at the end of the preference lists will be eliminated in the first $(t-2)$ stages, and the candidates who are positioned at the front will be all included in the final winning committee. 
All other procedures remain the same.
\qed
\end{proof}

%
%

\begin{theorem}[\textbf{Consistency does not preserve in a rational multi-stage $(\beta,\vec{\gamma})$-rule}]
\label{thm: l1l1_consistency}
Let $t\geq 2$ be an integer, $\beta\in\{\ell_1, \ell_{\max}\}$, and $\mathcal{R}$ be a rational $t$-stage $(\beta,\vec{\gamma})$-rule. 
Then $\mathcal{R}$ does not satisfy Consistency.
\end{theorem}
%
%
%
%

\noindent According to Skowron's work~\cite{skowron2019axiomatic}, only the scoring rules that generate weak linear orders over the committees exhibit Consistency. However, the multi-stage rules examined in this study fail to produce linear orders. This is due to the greedy nature of the multi-stage process, where candidates are eliminated at each stage based on their current contribution to the committee rather than their overall value. As a result, it cannot be asserted that the final winning committee necessarily outperform a committee consisting of candidates eliminated at different stages. Therefore, non-linear order multi-stage rules are deemed inconsistent.

\begin{proof}[Proof of Theorem~\ref{thm: l1l1_consistency}]
We present below a two-stage election that works as a counter-example for both $(\ell_1,\vec{\gamma})$-rules and $(\ell_{\max},\vec{\gamma})$-rules.
%
There are five candidates $a, b, c, d,$ and $e$.
The target size in each stage is represented by $\Vec{v} = (4, 1)$.
With rationality of $\vec{\gamma}$, we have at least one of the following inequalities holds: $\gamma^{4, 1}(1)>\gamma^{4, 1}(2), \gamma^{4, 1}(2)>\gamma^{4, 1}(3)$, or $\gamma^{4, 1}(3)>\gamma^{4, 1}(4)$. 
\begin{itemize}
    \item $\gamma^{4, 1}(1)>\gamma^{4, 1}(2)$. We construct $V_1$ as below:
    \begin{align*}
        &\text{Group 1:}\quad 1\times\ \rho_i(a,b,e)\succ c\succ d\quad \forall i\in\{1,2,3\};\\
        &\text{Group 2:}\quad 1\times\ \rho_i(b,a,c)\succ d\succ e\quad \forall i\in \{1,2,3\};\\
        &\text{Group 3:}\quad 300\times\ \pi\succ c\quad \forall \pi\in \Pi(a, b, d, e);\\
        &\text{Group 4:}\quad 100\times\ \pi\succ d\quad \forall \pi\in \Pi(a, b, c, e);\\
        &\text{Group 5:}\quad 400\times\ \pi\succ e\quad \forall \pi\in \Pi(a, b, c, d).
    \end{align*}
\noindent We construct Groups 3 - 5 to ensure that $e$ is knocked out in the first stage and the final winner cannot be $c$ or $d$.
%
%
%
It can be observed that the scores of $a$ and $b$ in Groups 2 - 5 are the same in both stages.
%
%
In Group 1, after $e$ is eliminated, the score of $a$ is $2\gamma^{(4, 1)}(1)+\gamma^{(4, 1)}(2)$, and that of $b$ is $\gamma^{(4, 1)}(1)+\gamma^{(4, 1)}(2)$.
%
Therefore, $a$ is the final winner.
%
Further, we construct $V_2$ similarly as follows:
\begin{align*}
    &\text{Group 1:}\quad 1\times\ \rho_i(a,b,d)\succ c\succ e\quad \forall i\in \{1,2,3\};\\
    &\text{Group 2:}\quad 1\times\ \rho_i(b,a,c)\succ d\succ e\quad \forall i\in \{1,2,3\};\\
    &\text{Group 3:}\quad 300\times\ \pi \succ c\quad \forall \pi\in \Pi(a, b, d, e);\\
    &\text{Group 4:}\quad 400\times\ \pi \succ d\quad \forall \pi\in \Pi(a, b, c, e);\\
    &\text{Group 5:}\quad 100\times\ \pi \succ e\quad \forall \pi\in \Pi(a, b, c, d).
\end{align*}
A similar analysis to that of $V_1$ indicates that $d$ is knocked out in the first stage, and the winner of the second stage is also $a$. 

However, upon a combination of $V_1$ and $V_2$, the candidate eliminated in the first stage shall be transferred to $c$, as the total number of votes ranking $c$ last is greater than that of $d$ and $e$; specifically, $300\times 2 > 400+100$.
%
After $c$ is eliminated, the scores of $a$ and $b$ are the same in Groups 1, 3, 4, and 5 in both $V_1$ and $V_2$.
In Group 2, the score of $b$ is $2\gamma^{4, 1}(1)+\gamma^{4, 1}(2)$ in the second stage, and that of $a$ is $\gamma^{4, 1}(1)+2\gamma^{4, 1}(2)$.
Thus the final winner becomes $b$. 

%
\item $\gamma^{4, 1}(2)>\gamma^{4, 1}(3)$ and $\gamma^{4, 1}(3)>\gamma^{4, 1}(4)$. We modify our example election by letting Group 1 of $V_1$ be $1\times c\succ \rho_i(a, b, e)\succ d\ \forall i\in \{1,2,3\}$.
%
Through such modification, $a$ gets a score of $2\gamma^{4, 1}(2)+\gamma^{4, 1}(3)$ and $b$ gets $\gamma^{4, 1}(2)+2\gamma^{4, 1}(3)$ in the second stage, so the winner is $a$.
Note that although $c$ performs better than $a$ and $b$ in Group 1, it cannot be the winner for the existence of Group 3.
We can implement a similar modification in Group 2 in $V_1$ and Groups 1 - 2 in $V_2$ to obtain a counter-example for the case of $\gamma^{4, 1}(2)>\gamma^{4, 1}(3)$.
The case of $\gamma^{4, 1}(3)>\gamma^{4, 1}(4)$ can also be solved with similar modification.
\qed
\end{itemize}

%
%
\end{proof}
\section{Discussion on Voting Rules Selection}
\label{sec:discussion}
This study's findings provide guidance for the selection of rules in multi-stage voting scenarios.
When seeking to choose a multi-stage score-based rule, some criteria such as Committee Monotonicity, Candidate Monotonicity, and Consistency should not be considered due to the absence of rules that satisfy them.
However, if the axiom of Solid Coalitions is considered important, a single-stage rule that satisfies Solid Coalitions (such as SNTV) can be employed to construct a multi-stage rule.

From the perspective of axioms, those of monotonicity, which requires that an increase in support for some elected candidates should not result in their knock-out, are unlikely to preserve in multi-stage voting.
The support for candidates may be increased by changing a voter's preference (Candidate Monotonicity) or by adding a set of voters in favor of them (Consistency).
The reason why the axioms of monotonicity usually do not preserve is that increasing support for certain candidates in these ways can change the outcome of the election for other candidates.
While the candidates receiving more support may still be elected by a single-stage rule, it is possible for the new members of this stage's winning committee to surpass them in the subsequent stages.

Notice that there is another commonly used family of multi-winner voting rules, called approval-based rules, under which voters are identified with their approval ballots. 
(See Appendix for formal definitions.)
In the case of approval-based rules, the axioms of Candidate Monotonicity and Consistency 
need not be taken into account. 
Basically, in multi-stage voting, approval-based rules are less prone to not preserving the desirable axioms as compared to score-based ones, specifically, Committee Monotonicity is preserved for approval-based rules but not for score-based rules.\footnote{Note that this observation may not hold if we take more general class of approval-based rules into consideration. For example, the “non-standard" rule of \emph{satisfaction approval voting} defines that the score of a committee $S$ given by a voter $v$ not only depends on $|A_v \cap S|$ but also relates to the size $|A_v|$ itself. The score can change a lot across stages as the candidate pool shrinks, and therefore the axioms are unlikely to be preserved in multi-stages.}
This is because, for score-based rules, a voter's satisfaction with regard to a committee (represented by the score of the committee given by the voter) exhibits less consistency across multiple stages.
For example, it is likely that a voter assigns higher scores to the remaining candidates in the subsequent stages as a result of the elimination of the candidate that she initially favors.
By contrast, in an approval-based voting, each voter's satisfaction with a committee solely depends on the selected number of her approved candidates, which does not change over different stages.

While our framework encompasses a wide range of multi-stage voting rules, it is crucial to recognize its limitations. 
Certain procedures, such as Cup and Black's method \cite{narodytska2013manipulating}, fall outside the scope of our framework. 
The Cup rule entails representing candidates as leaf nodes in a binary tree and comparing them pairwise until a unique winner is identified. 
On the other hand, Black's procedure first checks for the presence of a Condorcet winner. 
If such a candidate exists, it was presented as the winner of this election. Otherwise, the winner is determined based on the Borda rule.
Additionally, Nanson's rule \cite{Niou1987note} bears resemblance to our definition of multi-stage, expect that in each stage, candidates are eliminated based on whether their score falls below the average Borda score of all candidates.
Therefore, it is not possible to pre-define the number of remaining candidates at each stage, i.e., the vector $\vec{v}$.
Although these methods are notable in the field of multi-stage voting, they are not covered by our framework.


To summarize, the research presented in this paper establishes a foundation for the theoretical investigation of multi-stage voting and opens up several avenues for future exploration. 
One possible direction involves conducting additional axiomatic studies for other rules in multi-stages, such as Moreno's rules, greedy CC, and similar approaches, which were not covered in this paper. 
Furthermore, an interesting line of inquiry is to explore how the introduction of multi-stages can mitigate various manipulation techniques. 
As noted in Appendix~\ref{sec:emp}, multi-stage voting shows potential fairness advantages, warranting further theoretical analysis on whether multi-stage voting can effectively be employed for debiasing purposes.

\bibliographystyle{splncs04}
\bibliography{mybibliography}

\appendix
\clearpage

\section{Axiomatic Study for Multi-Stage Approval-Based Rules}
\label{sec:approval}


\subsection{Approval-Based Rules and Related Axioms}
\label{sec:approval_based_rules}

When voters are identified with their approval ballots, we refer to the voting method for selecting committees as the approval-based rule.
This appendix focuses on a quite general class of approval-based rules, known as \emph{Thiele methods}, introduced by~\cite{Thiele}.
There has been extensive research on this class of rules and its special cases in the social choice community (e.g., ~\cite{brill2018multiwinner,chamberlin1983representative,sornat2022near}).
The class of Thiele methods consists of all the rules that maximize the sum of voters' individual satisfaction, where each voter $v$'s satisfaction with committee $S$ solely depends on the number of $v$'s approved candidates in $S$.
For each voter $v\in V$, we denote by $A_{v}\subseteq C$ voter $v$'s approval ballot, i.e., the set that consists of the first $|A_v|$ candidates in her preference list.

\begin{definition}[\bf{Thiele Methods, $\omega$-Thiele}]
\label{def: thiele}
    A Thiele method is parameterized by a nondecreasing function $\omega:\mathbb{N}\rightarrow \mathbb{R}$ with $\omega(0) = 0$. The score of a committee $S$ given a set $V$ of voters is defined as
    $score_{\omega}(V,S) = \sum_{v\in V} \omega(|S \cap A_v|)$.
    The $\omega$-Thiele method returns committees with maximum scores.
\end{definition}
\noindent With different functions $\omega$, $\omega$-Thiele can cover a wide range of approval-based rules. 
An example is the most natural rule, \emph{Approval Voting} (AV for short), which selects the $k$ candidates that are approved by most voters.
AV is the $\omega$-Thiele with $\omega(x) = x$.
Another example is \emph{Proportional Approval Voting} (PAV for short), which is the $\omega$-Thiele with $\omega(x) = \sum^x_{j=1}\frac{1}{j}$.
The function $\omega$ depending on the sequence of harmonic numbers captures the property of diminishing returns.
We refer to a rule $\mathcal{R}$ as a multi-stage $\omega$-Thiele if it employs the same $\omega$-Thiele method with the same approval ballots $A_v$'s in each stage.

\paragraph{\textbf{Axioms Related to Approval-Based Rules.}} We present below some axioms that may be desirable for approval-based rules.
First, we consider \emph{Pareto Efficiency} in which a dominated committee must never be output.
When the goal of a multi-winner rule is to select the “best” committee, Pareto Efficiency is often considered to be a minimal requirement.\footnote{Pareto Efficiency can also be defined for score-based rules. However, it is not easy to extend the notion of Pareto Efficiency to multi-stage. Hence, we only consider Pareto Efficiency for approval-based rules in this paper.}
\begin{definition}[\textbf{Pareto Efficiency}]
\label{def:PE}
    A committee $S_1$ dominates a committee $S_2$ if (\romannumeral1) every voter has at least as many approved candidates in $S_1$ as in $S_2$ (for $v\in V$ it holds that $|A_v \cap S_1| \geq |A_v \cap S_2|$), and (\romannumeral2) there is one voter with strictly more approved candidates (there exists $u\in V$ with $|A_u \cap S_1| > |A_u \cap S_2|$).
    An approval-based rule $R$ satisfies Pareto Efficiency if $R$ never outputs dominated committees.
    \end{definition}
\noindent For approval-based rules, the definition of Committee Monotonicity is identical to that for score-based rules (Definition~\ref{def: comm_mono_multi}).
Our next axiom is \emph{Justified Representation}, which is a criterion for accessing whether an approval-based rule can be considered proportional.
The intuition of Justified Representation is that if $k$ candidates are to be selected, then a set of $\frac{n}{k}$ voters that are completely unrepresented can demand that at least one of their all-approved candidates be selected.
\begin{definition}[\textbf{Justified Representation}]
\label{def:JR}
    An approval-based rule $R$ satisfies Justified Representation iff $\forall E=(C, V)$ and $k\in[m]$, for any committee $S\in R(E, k)$, there does not exist a set of voters $V^*\subseteq V$ with $\left\|V^*\right\| \geq \frac{n}{k}$ such that $\bigcap_{v\in V^*} A_v \neq \varnothing$ and $A_v \cap S = \varnothing$ for all $v \in V^*$.
    %
\end{definition}
\noindent Candidate Monotonicity and Consistency should also be considered in approval-based voting when the fair treatment of candidates is necessary. 
\begin{definition}[\bf{Candidate Monotonicity for Approval-Based Rules}]
\label{def:candidate_monotonicity_approval}
An approval-based rule $R$ is candidate monotone iff $\forall E=(C, V)$, $k\in[m]$ and $c\in C$, if $c\in S$ for some $S\in R(A, k)$, then if a voter $v$ additionally approves the candidate $c$ and we obtain an election $E'$, we have $c\in S'$ for some $S'\in R(E', k)$.
\end{definition}
\noindent In Definition~\ref{def:candidate_monotonicity_approval}, the support of a candidate increases by a voter approving the candidate additionally, instead of shifting her one position forward (in Definition~\ref{def:candidate_monotonicity}).
On the other hand, Consistency for approval-based rules is defined in the same way as in Definition~\ref{def:consistency}.

Next, we provide an axiomatic study for approval-based rules.
We focus on the case when single-stage rules satisfy a certain desirable property and investigate whether the property can be preserved in multi-stages that apply one of the single-stage rules in each stage.
\subsection{Preserving Axioms in Multi-Stages}
\label{sec:approval_preserve}
We find the axioms of 
Committee Monotonicity (Theorem~\ref{thm:CM_approval}) and Justified Representation (Theorem~\ref{thm:JR}) preserved in multi-stage approval-based voting.

Committee Monotonicity is a very demanding criterion for Thiele methods.
We show that only the $\omega$-Thiele methods with linear $\omega$-functions satisfy Committee Monotonicity (Lemma~\ref{lem:CM_approval}).
If a multi-stage rule $\mathcal{R}$ is composed of rules that maximize\\$\sum_{v\in V} |A_v \cap S|$ over possible committees $S$ in each stage, it must finally select the candidates approved by most voters after multiple stages.
In fact, the multi-stage rule $\mathcal{R}$ produces a ranking of candidates and thus satisfies Committee Monotonicity by definition.

    \begin{theorem}[\textbf{Committee Monotonicity preserves in a multi-stage approval-based rule}]
    \label{thm:CM_approval}
    Let $t \geq 2$ be an integer and $\mathcal{R} = (R_1, R_2, \dots, R_t)$ be a $t$-stage approval-based rule. If $R_r$ is a single-stage $\omega$-Thiele rule that satisfies Committee Monotonicity for each $r\in[t]$, $\mathcal{R}$ also satisfies Committee Monotonicity.
    \end{theorem}

    \noindent
    Before proof of Theorem~\ref{thm:CM_approval}, we would like to first clarify what single-stage committee monotone rules are like.
    The following lemma indicates that only the $\omega$-Thiele methods with linear $\omega$-functions satisfy Committee Monotonicity. 
    \begin{lemma}
    \label{lem:CM_approval}
    If $R$ is a single-stage $\omega$-Thiele rule that satisfies Committee Monotonicity, then 
    $$\omega(i) - \omega(i-1) = \omega(j) - \omega(j-1)$$
    holds for any $i, j \in [1, \infty)$.
    \end{lemma}
    \begin{proof}
    Let $p_i$ denote the difference between $\omega(i)$ and $\omega(i-1)$ for each $i \geq 1$.
    Suppose for the sake of contradiction that there exists a nondecreasing function $\omega$ with $p_{i_0} \neq p_{1}$ for some $i_0 > 1$ that parameterizes a rule $R$ satisfying Committee Monotonicity. Without loss of generality, we assume $i_0$ is the smallest integer such that $p_{i_0} \neq p_{1}$. There are two cases:
    \begin{itemize}
        \item $p_{i_0} < p_{1}$. There is a rational number $\frac{n_1}{n_2}$ such that $\frac{p_{i_0}}{p_1} < \frac{n_1}{n_2} < 1$, where $n_1$ and $n_2$ are both integers. Let us consider the voters with approval ballots as below:
        \begin{equation*}
            \begin{split}
            n_1 \times \{a\}\quad n_2 \times \{a, c_1, c_2, \dots, c_{i_0-1}\}\\
            n_1 \times \{b\}\quad n_2 \times \{b, c_1, c_2, \dots, c_{i_0-1}\}. 
            \end{split}
        \end{equation*}
        $R$ chooses $\{c_1, c_2, \dots, c_{i_0-1}\}$ for $k = i_0-1$. However, for $k = i_0$, $R$ chooses $\{a, b, c_{l1}, c_{l2}, \dots, c_{l,i_0-2}\}$ as winning committee, where $c_{l1},...,c_{l,i_0-2}$ are $(i_0-2)$ candidates chosen arbitrarily from $\{c_1,...,c_{i_0-1}\}$. The only winning committee in $R(E, i_0-1)$ is not a subset of $\{a, b, c_{l1}, ..., c_{l,i_0-2}\}$, contradicting Committee Monotonicity of $R$.\\
        
        \item $p_{i_0} > p_{1}$. There is a rational number $\frac{n_1}{n_2}$ such that $\frac{p_{i_0}}{p_1} > \frac{n_1}{n_2} > 1$. Consider an election $E$ with voters described below:
        \begin{align*}
            n_1 \times \{c, d_1, ..., d_{i_0-2}\} \quad n_2 \times \{a, b, d_1, ..., d_{i_0-2}\}.
        \end{align*}
        $R$ chooses $\{c, d_1, ..., d_{i_0-2}\}$ for $k = i_0-1$ but chooses $\{a, b, d_1,..., d_{i_0-2}\}$ for $k = i_0$. There is no committee $S$ in $R(E, i_0)$ such that $\{c, d_1, \dots, d_{i_0-2}\} \subseteq S$, contradicting Committee Monotonicity of $R$.
    \end{itemize}
    \qed
    \end{proof}
    \noindent With Lemma~\ref{lem:CM_approval}, we can now complete proof of Theorem~\ref{thm:CM_approval}.
    \begin{proof}[Proof of Theorem~\ref{thm:CM_approval}]
    For each $r\in [t]$, the score of a committee $S$ in $R_r$ is equivalently defined as $\sum_{v\in V} |A_v\cap S|$ by Lemma~\ref{lem:CM_approval}. 
    In each stage, $R_r$ selects the candidates that are approved by most voters. 
    Then the $t$-stage rule $\mathcal{R}=(R_1, \dots, R_t)$ finally selects $k$ candidates approved by most voters.
    $\mathcal{R}$ actually produces a ranking of candidates and satisfies Committee Monotonicity by definition.
    \qed
    \end{proof}
There are many single-stage rules that fail to satisfy the criterion of Justified Representation, though it is considered to be one of the less stringent definitions of “proportional representation" among several possible options.
An example is AV, mentioned in Section~\ref{sec:approval_based_rules}, which simply selects the candidates approved by most voters~\cite{aziz2017justified}.
Despite this, the good news is that we can simply prove by contradiction that a multi-stage rule satisfies Justified Representation if it is composed of single-stage $\omega$-Thiele methods that satisfy Justified Representation.

    \begin{theorem}[\textbf{Justified Representation preserves in a multi-stage approval-based rule}]
    \label{thm:JR}
    Let $t \geq 1$ be an integer and $\mathcal{R} = (R_1,R_2,...,R_t)$ be a $t$-stage multi-winner approval-based voting rule. If $R_r$ is a single-stage $\omega$-Thiele rule that satisfies Justified Representation for each $r\in[t]$, $\mathcal{R}$ also satisfies Justified Representation.
    \end{theorem}
    \begin{proof}
    Fix an election $E=(C, V)$ and a $\Vec{v} = (k_1, ...,k_{t})$. 
    Let $S$ be one of the winning committees in the output of $\mathcal{R}=(R_1,\dots,R_t)$ on $(E, \Vec{v})$. 
    Suppose for the sake of contradiction that $R_r$ satisfies Justified Representation for each $r\in[t]$, but 
    there exists a set $V^*\subseteq V$ with $\left\|V^*\right\| \geq \lceil\frac{n}{k}\rceil$ such that $\bigcap_{v\in V^*}A_v \neq \varnothing$ and $\left(\bigcup_{v\in V^*}A_v\right) \cap S = \varnothing$.

    For each $r\in[t]$, let $S_r$ be the output of $R_r$.
    Since $\left(\bigcup_{v\in V^*}A_v\right) \cap S = \varnothing$ and $R_t$ satisfies Justified Representation by assumption, $\left(\bigcap_{v\in V^*}A_v\right) \cap S_{t-1}$ must be $\varnothing$.
    We have $S = S_t\subset\cdots\subset S_1\subset S_0 = C$, $\left(\bigcap_{v\in V^*}A_v\right) \cap S_{t-1} = \varnothing$ and $\left(\bigcap_{v\in V^*}A_v\right) \neq \varnothing$. Let $r_0 \in [1, t)$ be the last stage such that $\left(\bigcap_{v\in V^*}A_v\right) \cap S_{r_0-1} \neq \varnothing$, i.e., $\left(\bigcap_{v\in V^*}A_v\right) \cap S_{r_0} = \varnothing$. 
    This means the candidates in $\left(\bigcap_{v\in V^*}A_v\right) \cap S_{r_0-1}$ are not included in the winning set $S_{r_0}$. There are $\left\|V\right\|\geq \frac{n}{k_{r_0}}$ voters unrepresented in the $r_0$-th stage, which contradicts Justified Representation of $R_{r_0}$.
    \qed
    \end{proof}

\subsection{Violating Axioms in Multi-Stages}
\label{sec:approval_violate}

    We discuss in this part the axioms of Candidate Monotonicity (Theorem~\ref{thm:candidate_monotone_approval}), Consistency (Theorem~\ref{thm:consistency_approval}), and Pareto Efficiency (Theorem~\ref{thm:PE}).
    %
    It is known that all of the single-stage increasing Thiele methods satisfy Candidate Monotonicity, Consistency, and Pareto Efficiency (e.g., see~\cite{lackner2020multi}).
    %
    However, a multi-stage approval-based rule may not satisfy them even though it is composed of single-stage Thiele methods.

    \begin{theorem}[\textbf{Candidate Monotonicity does not preserve in a multi-stage approval-based rule}]
    \label{thm:candidate_monotone_approval}
    Let $t\geq 2$ be an integer and $\mathcal{R}$ be a $t$-stage $\omega$-Thiele rule, where $R$ is a single-stage $\omega$-Thiele rule that satisfies Candidate Monotonicity.
    If for any integer $i \geq 0$, $\omega(i)$ is a rational number  and
    $$0\neq\omega(i_0-1) - \omega(i_0-2) \neq \omega(i_0) - \omega(i_0-1)$$
    holds for some $i_0$ in the function $\omega$, the multi-stage rule $\mathcal{R}$ is not candidate monotone.
    \end{theorem}

    \noindent
    Candidate Monotonicity requires that an increase in support for some elected candidate by changing a voter's preference should not result in their knock-out.
    The reason why this axiom does not preserve is that changing some votes in favor of an elected candidate can change the outcome of the election for other candidates.
    While the candidate receiving more support may still be elected by a single-stage rule, it is possible for the new members of this stage's winning committee to surpass them in the subsequent stages.
    
    \begin{proof}[Proof of Theorem~\ref{thm:candidate_monotone_approval}]
    For the function $\omega$ parameterizing $R$, let $p_i$ denote the difference between $\omega(i)$ and $\omega(i-1)$ for each $i \geq 1$.
    There are two cases: $p_{i_0} < p_{i_0-1}$ and $p_{i_0} > p_{i_0-1}$. Our plan is to construct an election $E=(C, V)$ and a vector $\Vec{v} = (k_1,k_2,\dots,k_t)$ for each case to show that $\mathcal{R}$ does not satisfy Candidate Monotonicity.
    \begin{itemize}
        \item $p_{i_0} < p_{i_0-1}$. We can always find a rational number $\frac{n_1}{n_2}$ such that $n_1$ and $n_2$ are both integers, and $\frac{n_1}{n_2} = \frac{p_{i_0-1}}{p_{i_0}}$.
        Without loss of generality, we assume that $n_1 - n_2 > 1$.
        Otherwise, we can use $c\cdot n_1$ and $c\cdot n_2$ for some integer $c$ instead.
        Then consider the following set of voters $V$:
        \begin{equation*}
        \begin{split}
            & n_1 \times \{a, c, d_1, \dots, d_{i_0-2}\}\quad n_1 \times \{b, c, d_1, \dots, d_{i_0-2}\}\\
            & n_2 \times \{a, d_1, \dots, d_{i_0-2}\}\quad n_2 \times \{b, d_1, \dots, d_{i_0-2}\}\\
            & 1 \times \{e_{1}\}\quad 1 \times \{e_{2}\}\quad \cdots \quad 1 \times \{e_{t-2}\}.
        \end{split}
        \end{equation*}
        We eliminate one candidate in each stage. In particular, $\mathcal{R}$ selects $i_0$ candidates from $(i_0+1)$ ones in the $(t-1)$-th stage, and picks $(i_0-1)$ candidates from $i_0$ ones in the last stage. 
        
        It is straightforward that candidate $e_1,...,e_{t-3}$ and $e_{t-2}$ are eliminated in the first $(t-2)$ stages. In the $(t-1)$-th stage, $\{a,b, d_1, ..., d_{i_0-2}\}$ 
        is one of the winning committees. Its score is equal to that of $\{a,c, d_1, ..., d_{i_0-2}\}$ or $\{b, c, d_1, ..., d_{i_0-2}\}$, as $n_1p_{i_0} = n_2p_{i_0-1}$.
        %
        Then there must be a winning committee $S$ in the last stage such that $b\in S$.
        %
        However, if a voter who has approved $\{a, d_1, \dots, d_{i_0-2}\}$ additionally approves $b$, then $\{b, c, d_1, ..., d_{i_0-2}\}$ becomes the only winning committee in the $(t-1)$-th stage, which outperforms $\{a,b, d_1, ..., d_{i_0-2}\}$.
        Hence the final winning committee is $\{c, d_1, \dots, d_{i_0-2}\}$, which does not include $b$.
        
        \item $p_{i_0} > p_{i_0-1}$. There exists a rational number $\frac{n_1}{n_2}<1$ such that $n_1$ and $n_2$ are both integers and $\frac{n_1}{n_2} = \frac{p_{i_0-1}}{p_{i_0}}$. Without loss of generality, we assume that $n_2 - n_1 > 1$.
        Then consider voters with the following approval ballots:
        \begin{equation*}
        \begin{split}
            & n_1 \times \{a, b, d_1, \dots, d_{i_0-2}\}\quad n_2 \times \{c, d_1, \dots, d_{i_0-2}\}\\
            & 1 \times \{e_{1}\}\quad 1 \times \{e_{2}\}\quad \cdots \quad 1 \times \{e_{t-2}\}.
        \end{split}
        \end{equation*}
        We still consider an election that eliminates one candidate in each stage. Candidates $e_1,...,e_{t-2}$ are eliminated in the first $(t-2)$ stages. In the $(t-1)$-th stage, $\{a,b, d_1, ..., d_{i_0-2}\}$ is one of the winning committees.
        %
        Then there is a winning committee $S$ in the last stage such that $b\in S$. However, if a voter who has approved $\{c, d_1, \dots, d_{i_0-2}\}$ additionally approves $b$, then $\{b, c, d_1, ..., d_{i_0-2}\}$ becomes the only winning committee in the $(t-1)$th-stage, which outperforms $\{a,b, d_1, ..., d_{i_0-2}\}$. 
        Hence the final winning committee is $\{c, d_1, \dots, d_{i_0-2}\}$ as $n_2 > n_1+1$. The only winning committee does not include $b$.
    \end{itemize}
    \qed
    \end{proof}

    \begin{theorem}[\textbf{Consistency does not preserve in a multi-stage approval-based rule}]
    \label{thm:consistency_approval}
    Let $t\geq 2$ be an integer and $\mathcal{R}$ be a $t$-stage $\omega$-Thiele rule, where $R$ is a single-stage $\omega$-Thiele rule that satisfies Consistency. If 
    $$\omega(i_0-1) - \omega(i_0-2) \neq \omega(i_0) - \omega(i_0-1)$$
    holds for some $i_0$ in the function $\omega$, the multi-stage rule $\mathcal{R}$ does not satisfy Consistency.
    \end{theorem}

    \noindent
    The reason why Consistency does not inherit is similar to that of Candidate Monotonicity.
    Increasing support for certain candidates by adding a set of voters in favor of them not only affects the candidates themselves with support but also alters the outcome for other candidates.
    Although the candidates receiving additional support may still be elected at first, it is possible for the new members of the winning committee in this stage to knock out them in the subsequent stages.
    Consistency therefore cannot preserve in multi-stage voting.
    
    \begin{proof}[Proof of Theorem~\ref{thm:consistency_approval}]
    For the function $\omega$ parameterizing $R$, let $p_i$ denote the difference between $\omega(i)$ and $\omega(i-1)$ for each $i \geq 1$.
    Without loss of generality, we assume $i_0$ is the smallest integer such that $p_{i_0} \neq p_{1}$. 
    There are two cases: $p_{i_0} < p_{1}$ and $p_{i_0} > p_{1}$. We plan to construct an election $E=(C, V)$ and a vector $\vec{v} = (k_1,k_2,\dots,k_t)$ for each case to show that $\mathcal{R}$ does not satisfy Consistency.
    \begin{itemize}
        \item $p_{i_0} < p_{1}$. 
        Consider an election that eliminates one candidate in each stage, i.e., $\Vec{v} = (i_0+t-2, i_0+t-3, ..., i_0, i_0-1)$.
        There is a rational number $\frac{n_1}{n_2}$ such that $\frac{p_{i_0}}{p_1} < \frac{n_1}{n_2} < 1$ and $n_1 > 2$. Let us consider the set of voters $V_1$ with approval ballots as below:
        \begin{equation*}
            \begin{split}
            & n_1 \times \{a\}\quad n_2 \times \{b, c_1, c_2, \dots, c_{i_0-1}\}\\
            & 1 \times \{e_{1}\}\quad 1 \times \{e_{2}\}\quad \cdots \quad 1 \times \{e_{t-2}\}.
            \end{split}
        \end{equation*}
        It is obvious that candidates $e_1,...,e_{t-2}$ are eliminated in the first $(t-2)$ stages. In the $(t-1)$-th stage, $\mathcal{R}$ may choose $\{a, c_{1}, \dots, c_{i_0-1}\}$ or $\{a, b, c_{l1}, c_{l2}, \dots, c_{l,i_0-2}\}$ as winning committee, where $c_{l1}, ..., c_{l, i_0-2}$ are arbitrary $(i_0-2)$ candidates chosen from $c_1, ..., c_{i_0-1}$. In the last stage, $a$ will be eliminated due to that $n_2 > n_1$. Hence the final winning committee may be $\{c_{1}, \dots, c_{i_0-1}\}$ or $\{b, c_{l1}, c_{l2}, \dots, c_{l,i_0-2}\}$.

        We construct the set of voters $V_2$ as below:
        \begin{equation*}
            \begin{split}
            & n_1 \times \{b\}\quad n_2 \times \{a, c_1, c_2, \dots, c_{i_0-1}\}\\
            & 1 \times \{e_{1}\}\quad 1 \times \{e_{2}\}\quad \cdots \quad 1 \times \{e_{t-2}\}. 
            \end{split}
        \end{equation*}
        Similarly, candidates $e_1,...,e_{t-2}$ are eliminated in the first $(t-2)$ stages. In the $(t-1)$-th stage, $\mathcal{R}$ may choose $\{b, c_{1}, \dots, c_{i_0-1}\}$ or $\{a, b, c_{l1}, c_{l2}, \dots, c_{l,i_0-2}\}$ as winning committee, where $c_{l1}, ..., c_{l, i_0-2}$ are arbitrary $(i_0-2)$ candidates chosen from $c_1, ..., c_{i_0-1}$. In the last stage, $b$ will be eliminated due to that $n_2 > n_1$. Hence the final winning committee may be $\{c_{1}, \dots, c_{i_0-1}\}$ or $\{a, c_{l1}, c_{l2}, \dots, c_{l,i_0-2}\}$. Therefore $\mathcal{R}((C,V_1), \vec{v}) \cap \mathcal{R}((C,V_2), \vec{v}) = \left\{\{c_{1}, \dots, c_{i_0-1}\}\right\}$.

        However, if we combine $V_1$ and $V_2$, $\mathcal{R}$ will choose $\{a, b, c_{l1}, \dots, c_{l,i_0-2}\}$ as winning committee in the $(t-1)$-th stage since $p_1n_1 > p_{i_0}n_2$. In the last stage, committee $\{c_{1}, \dots, c_{i_0-1}\}$ must not be in $\mathcal{R}((C,V_1+V_2), i_0-1)$, and thus $\mathcal{R}$ does not satisfy Consistency.

        \item $p_{i_0} > p_1$. Let $\Vec{v} = (i_0+t-2,i_0+t-3,...,i_0, i_0-1)$, i.e., one candidate is eliminated in each stage.
        There is a rational number $\frac{n_1}{n_2}$ such that $\frac{p_{i_0}}{p_1} > \frac{n_1}{n_2} > 1$ and $n_2 > 2$. Consider the set of voters $V_1$ with approval ballots described below:
        \begin{equation*}
            \begin{split}
            & n_1 \times \{b, d_1, \dots, d_{i_0-2}\} \quad n_1\times \{c, d_1, \dots, d_{i_0-2}\}\\
            & n_2 \times \{a, b, d_1, \dots, d_{i_0-2}\}\\
            & n_2 \times \{a, c, d_1, \dots, d_{i_0-2}\}\\
            & 1 \times \{e_{1}\}\quad 1 \times \{e_{2}\}\quad \cdots \quad 1 \times \{e_{t-2}\}. 
            \end{split}
        \end{equation*}
        After candidates $e_1,e_2,...,e_{t-2}$ are eliminated in the first $(t-2)$ stages, $\mathcal{R}$ chooses $\{a, b, d_1, \dots, d_{i_0-2}\}$ or $\{a, c, d_1, \dots, d_{i_0-2}\}$ in the $(t-1)$-th stage. Then $a$ will be eliminated in the last stage as $n_2 < n_1$. Hence the final winning committee is $\{b, d_1, \dots, d_{i_0-2}\}$ or $\{c, d_1, \dots, d_{i_0-2}\}$.

        We construct a set of voters $V_2$ as below:
        \begin{equation*}
            \begin{split}
            & n_1 \times \{a, d_1, \dots, d_{i_0-2}\} \quad n_1\times \{c, d_1, \dots, d_{i_0-2}\}\\
            & n_2 \times \{a, b, d_1, \dots, d_{i_0-2}\}\\
            & n_2 \times \{b, c, d_1, \dots, d_{i_0-2}\}\\
            & 1 \times \{e_{1}\}\quad 1 \times \{e_{2}\}\quad \cdots \quad 1 \times \{e_{t-2}\}. 
            \end{split}
        \end{equation*}
        $\mathcal{R}$ chooses $\{a, b, d_{1}, \dots, d_{i_0-2}\}$ or $\{b, c, d_1, \dots, c_{i_0-2}\}$ in the $(t-1)$-th stage. In the last stage, $b$ will be eliminated, and thus the final winning committee is  $\{a, d_{1}, \dots, d_{i_0-2}\}$ or $\{c, d_1, \dots, d_{i_0-2}\}$. Hence $\mathcal{R}((C,V_1), \vec{v}) \cap \mathcal{R}((C,V_2), \Vec{v}) = \left\{\{c, d_{1}, \dots, d_{i_0-2}\}\right\}$.

        However, if we combine $V_1$ and $V_2$, $\mathcal{R}$ will choose $\{a, b, d_{1}, \dots, d_{i_0-2}\}$ as winning committee in the $(t-1)$-th stage since $p_{i_0}n_2 > p_1n_1$. In the last stage, committee $\{c, d_{1}, \dots, d_{i_0-2}\}$ must not be in $\mathcal{R}((C,V_1+V_2), \vec{v})$, and thus $\mathcal{R}$ does not satisfy Consistency.
    \end{itemize}
    \qed
    \end{proof}

\begin{theorem}[\textbf{Pareto Efficiency does not preserve in a multi-stage approval-based rule}]
    \label{thm:PE}
    Let $t\geq 2$ be an integer and $\mathcal{R}$ be a $t$-stage $\omega$-Thiele.
    If
    \begin{equation}
    \label{eq:pareto}
      \omega(i_0-1) - \omega(i_0-2) > \omega(i_0) - \omega(i_0-1) 
    \end{equation}
    holds for some $i_0$ in the function $\omega$, 
    then $\mathcal{R}$ does not satisfy Pareto Efficiency.
    \end{theorem}

     \noindent Note that if we consider multi-stage voting with vector $(m-1, m-2, m-3,\dots,k)$, we indeed reverse sequential Thiele rules. It is well-known that reverse sequential proportional approval voting does not satisfy Pareto Efficiency~\cite{lackner2020multi}. Theorem~\ref{thm:PE} generalizes the negative result to any $\omega$-Thiele rule which has some $i_0$ satisfying Ineq.~\ref{eq:pareto}.
    \begin{proof}[Proof of Theorem~\ref{thm:PE}]
        For the function $\omega$ parameterizing $\mathcal{R}$, let $p_i$ denote the difference between $\omega(i)$ and $\omega(i-1)$ for each $i \geq 1$.
        Our plan is to construct an election $E=(C, V)$ and a vector $\Vec{v} = (k_1,k_2,\dots,k_t)$ to show that $\mathcal{R}$ does not satisfy Pareto Efficiency.
        Consider an election that eliminates one candidate in each stage, i.e., $\Vec{v} = (i_0+t, i_0+t-1, ..., i_0+2, i_0+1)$.
        There exist integers $n_1, n_2, n_3 > 1$ satisfying $\frac{n_1}{n_2} \leq \frac{p_{i_0-1}-p_{i_0}}{p_{i_0}}$ and $\frac{n_3}{n_2} \geq \frac{p_{i_0-1}}{p_{i_0+1}}$. Let us consider the set of voters $V$ with approval ballots as below:
        \begin{equation*}
        \label{eq:profile_pareto_1}
        \begin{split}
            & n_1 \times \{a, f_1, ..., f_{i_0-2}\}\quad n_2 \times \{a, b, f_1, ..., f_{i_0-2}\} \quad n_2 \times \{a, c, f_1, ..., f_{i_0-2}\} \\ 
            & n_3 \times \{b,d,e, f_1, ..., f_{i_0-2}\}\quad  n_3 \times \{c,d,e, f_1, ..., f_{i_0-2}\}\\
            & 1 \times \{g_{1}\}\quad 1 \times \{g_{2}\}\quad \cdots \quad 1 \times \{g_{t-2}\}.
        \end{split}
        \end{equation*}
        It is straightforward that candidate $g_1,...,g_{t-3}$ and $g_{t-2}$ are eliminated in the first $(t-2)$ stages. It can be verified that in the $(t-1)$-th stage, one of the winning committees (of size $i_0+2$) is $\{b,c,d,e,f_1, ..., f_{i_0-2}\}$.
        In the last stage, we will choose either $\{b,c,d,f_1, ..., f_{i_0-2}\}$ or $\{b,c,e,f_1, ..., f_{i_0-2}\}$ as these have maximal score among all $(i_0+1)$-subsets of $\{b,c,d,e,f_1, ..., f_{i_0-2}\}$.  However, these two subsets are both dominated by $\{a,d,e,f_1,...,f_{i_0-2}\}$, so Pareto Efficiency fails.
        \qed


        %
    \end{proof}

\section{Empirical results}
\label{sec:emp}
We present a comparative experiment on single-stage and two-stage voting in this section, aiming to gain insight into the fairness of election outcomes.
We conduct simulations on synthetic data and employ a variety of commonly used score-based voting rules.
The empirical findings indicate that two-stage voting may result in a fairer selection of winning committees when compared to the single-stage method.
%
%
%
%
%

\subsection{Setup}
\paragraph{Sampling candidates and voters.}
We generate 200 random points on the plane $\mathbb{R}^2$ to represent candidates, where $80$ of them come from a Gaussian distribution centered at $(1,0)$ with a standard deviation of $0.5$, and the other $120$ ones are distributed uniformly in the square $[-2,1]\times[-2,1]$.
In addition, we uniformly sample 400 points on a disc centered at $(0,0)$ with a radius of $2$ to represent voters.
%
A voter's preference order is determined by the Euclidean distance between the voter and each candidate.
That is,
given a pair of candidates $c_i, c_j\in \mathbb{R}^2$ and a voter $v\in \mathbb{R}^2$, $v$ prefer $c_i$ to $c_j$ if $d(c_i,v)< d(c_j,v)$, where $d(\cdot, \cdot )$ stands for the Euclidean distance.

\paragraph{Target committee sizes.} For an election with two stages, we set the target committee size $k_2$ in the second stage as $20$.
%
%
Then the size $k_1$ of the winning committee for the first stage ranges between $20$ and $200$. 
We run a separate election for each value of $k_1$.
%
In particular, two-stage voting degenerates into single-stage voting when $k_1 = 20$ or $k_1 = 200$.
%
%
\paragraph{Voting rules.}
We apply the following score-based voting rules in our experiment: SNTV, Borda, Bloc, and CC.
The definitions of these rules can be found in Section \ref{sec:model}.
We implement the CC rule through integer linear programming (ILP) solving. 

\paragraph{The measure of fairness.}
%
We use the Gini index to evaluate the fairness of voting.
%
If $n_i$ represents the number of final winners in the $i$-th quadrant, then the Gini index $G$ is given by 
$G = \left(\sum^{4}_{i=1} \sum^{4}_{j=1} |n_i-n_j|\right)/\left(2 \sum^4_{i=1} \sum^{4}_{j=1}n_i\right)$.
%
The lower Gini index means a fairer election.

\subsection{Results}
Figure~\ref{fig_score_sntv} illustrates the mean and standard deviation of the scores of the winning committees, as determined through $500$ random trial elections applying SNTV in each stage.
Similarly, Figure~\ref{fig_gini_sntv} presents the mean and standard deviation of the Gini indices obtained from the same $500$ random trials using the SNTV rule.
%
%
%
The results of the remaining rules can be found in the appendix.
%
\begin{figure}[ht]
\centering
\includegraphics[height=6cm,width=8cm]{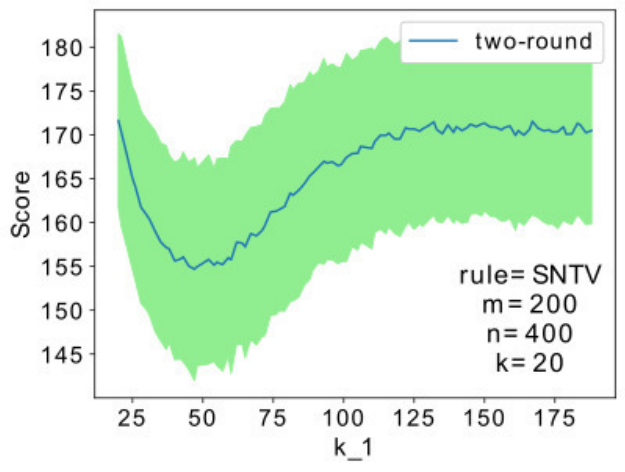}
\caption{Score of the winning committee under two-stage voting.
The blue line represents the score, and the green shade represents the standard deviation of the score.}
\label{fig_score_sntv}
\end{figure}

\paragraph{The winners produced by single-stage and two-stage voting can be distinct.}
As shown in Figure~\ref{fig_score_sntv}, the score curve displays a trend of initial decrease followed by an increase.
%
It is intuitive to infer that the score of the winning committee obtained through the two-stage voting process will be approximately equivalent to that obtained by a single-stage election as the value of $k_1$ approaches $k_2$ or $m$.
%
This is because, in such cases, the two-stage election degenerates into a single-stage one that chooses $20$ winners from $200$ candidates.
%
In particular, it can be observed that the score of the winning committee obtained through a two-stage voting process with a $k_1$ astage $50$ is lower than that obtained by a single-stage election on average.
This discrepancy in scores illustrates that the winners produced by single-stage and two-stage voting are different.
%
%
%


\begin{figure}[ht]
\centering
\includegraphics[height=6cm,width=8cm]{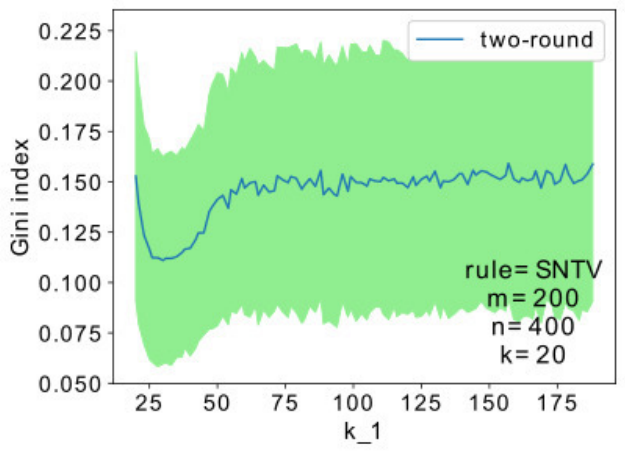}
\caption{Gini index of the winning committee under two-stage voting.
The blue line represents the Gini index, and the green shade represents the standard deviation of the Gini index.}
\label{fig_gini_sntv}
\end{figure}

\paragraph{Two-stage voting may be fairer than single-stage voting.}
%
%
Figure~\ref{fig_gini_sntv} shows that the Gini index of a two-stage voting process, with a value of $k_1$ approximately equal to $30$, is lower than that of a single-stage one on average.
This observation suggests that two-stage voting using SNTV can generate fairer winning committees than single-stage voting.

\end{document}